\documentclass[a4paper, 12pt]{article} 
\usepackage[utf8]{inputenc}
\usepackage[top=1.2in, bottom=1.2in, left=1.1in, right=1.1in]{geometry}
\usepackage{amsthm, amstext, amssymb, amsmath}
\usepackage{color, graphics}
\usepackage{url}
\usepackage{framed}
\usepackage{tikz}
\usepackage{float}
\usepackage{listings} 
\usepackage{algorithm2e}
\usepackage{cite}
\usepackage{xcolor}

\usetikzlibrary{matrix} 
\usetikzlibrary{shapes}
\tikzstyle{vertex}=[circle,draw=black,fill=white,minimum size=20pt,inner sep=0pt]
\tikzstyle{oedge} = [draw,thick,->]
\tikzstyle{edge} = [draw,thick]
\tikzstyle{selected edge} = [draw,line width=5pt,-,red!50]
\tikzstyle{weight} = [font=\small]

\newtheorem{theorem}{Theorem}
\newtheorem{lemma}[theorem]{Lemma}

\newtheorem{definition}[theorem]{Definition}
\newtheorem{observation}[theorem]{Observation}

\def\NPh{\textsf{NP}-hard}
\def\NP{\textsf{NP}}
\def\FPT{\textsf{FPT}}
\def\W{\textsf{W}}
\def\NPc{\textsf{NP}-complete}

\def\N{\mathbf{N}}

\def\to{\rightarrow}

\def\ss{\subseteq}

\def\x{\mathbf{x}}

\def\permpat{{\sc Permutation Pattern Matching}} 
\def\clique{{\sc Clique}} 
\def\language{\Sigma^*} 
\def\assumptioncoNP{$\mbox{\NP} \not\subseteq \mbox{co-\NP}/\mbox{poly}$}
\def\th{\textsuperscript{th}}
\def\O{\mathcal{O}}

\newcommand\term[1]{{\em #1}}

\begin{document}
\title{Kernelization lower bound for Permutation Pattern Matching\footnote{Partially supported by Warsaw Center of Mathematics and Computer Science.}}
\author{Ivan Bliznets\footnote{St.~Petersburg Department of Steklov Institute of Mathematics. E-mail: \texttt{iabliznets@gmail.com}.  Partially supported by the Government of the Russian Federation (grant 14.Z50.31.0030).}, Marek Cygan\footnote{Institute of Informatics, University of Warsaw, Poland. E-mail: \texttt{cygan@mimuw.edu.pl}. Partially supported by NCN grant DEC-2012/05/D/ST6/03214.}, Pawe\l\ Komosa\footnote{Institute of Informatics, University of Warsaw, Poland. E-mail: \texttt{kompaw01@gmail.com}.}, Luk\'a\v s Mach\footnote{DIMAP and Department of Computer Science, University of Warwick, United Kingdom. E-mail: \texttt{lukas.mach@gmail.com}. This author has also recieved funding from the European Research Council under the European Union's Seventh Framework Programme (FP7/2007-2013)/ERC grant agreement no.~259385.}} 
\maketitle 

\begin{abstract}
A permutation $\pi$ contains a permutation $\sigma$ as a pattern if it contains a subsequence of length $|\sigma|$ whose elements are in the same relative order as in the permutation $\sigma$. 
This notion plays a major role in enumerative combinatorics.
We prove that the problem does not have a polynomial kernel (under the widely believed complexity assumption \assumptioncoNP) by introducing a new polynomial reduction from the clique problem to permutation pattern matching. 
\end{abstract}

\section{Introduction} 

Counting permutations of size $n$ avoiding a fixed pattern is an established and active area of enumerative combinatorics. 
Knuth \cite{K} has shown that the number of permutations avoiding $(2, 3, 1)$ is the $n$\th\ Catalan number. 
Various choices of prohibited patterns have been studied among others by Lov\' asz \cite{L}, Rotem \cite{R}, and Simion and Schmid \cite{SS}.
This culminated in the Stanley-Wilf conjecture stating that for every fixed prohibited pattern, the number of permutations of length $n$ avoiding it can be bounded by $c^n$ for some constant $c$. 
Klazar \cite{Kstanley} reduced the question to the F\" uredi-Hajnal conjecture, which was ultimately proved by Marcus and Tardos in 2004 \cite{MT}. 

Wilf \cite{W} also posed the algorithmic question of whether detecting a given pattern (of length $\ell$) in a given permutation (of length $n$) can be done in subexponential time. 
Subsequently, the problem was shown to be \NPh\ in \cite{Lnp}.
Ahal and Rabinovich have obtained an $\O(n^{0.47\ell + o(\ell)})$ time algorithm \cite{A}.
Fast algorithms have been found for certain restricted versions of the problem \cite{S, I}.

Pattern matching has also received interest in the context of parameterized complexity. 
Several groups of researchers have obtained \W[1]-hardness results for generalizations of the problem \cite{Glin, M}.
In \cite{M2} it was shown that the problem is in \FPT\ when parameterized by the number of runs (maximal monotonic consecutive subsequences) in the target permutation.
The authors of \cite{M2} raise the issue of whether their problem has a polynomial size kernel as an open problem. 
The central question of whether the problem is in \FPT\ when parameterized by $\ell$ has been resolved by Guillemot and Marx \cite{Glin}, who obtained an algorithm with runtime of $2^{\O(\ell^2 \log \ell)} \cdot n$. 
This implies the existence of a kernel for the problem. 
Obtaining kernel size lower bounds was posed as an open question during a plenary talk at {\em Permutation Patterns 2013} by St\' ephane~Vialette. 

We prove that the permutation pattern problem under the standard parameterization by $\ell$ does not have a polynomial size kernel, assuming \assumptioncoNP.

\section{Preliminaries} 
The set $\{i, i+1, \ldots, j-1, j\}$ is denoted by $[i, j]$. 
We set $[n] := [1, n]$. 
A permutation $\pi$ is a bijection from $[n]$ to $[n]$.
The value $\pi(i)$ is called \term{the entry of $\pi$ at position $i$}.
We use $|\pi|$ to denote the size of the domain of $\pi$.
Two common representations of a permutation $\pi$ are used:
the vector $(\pi(1), \pi(2), \ldots, \pi(n))$ and the corresponding permutation matrix.
The latter is a $|\pi|\times|\pi|$ binary matrix with 1-entries precisely on coordinates $(\pi(i), i)$. 
A vector obtained from the vector representation by omitting some entries is \term{a subsequence of the permutation}. 
Such subsequence is \term{a consecutive subsequence} if it contains precisely the entries with indexes from $[i, j]$ for some $i, j \in \N$.
We use $\pi[i, j]$ to denote the set of entries $\{\pi(i), \pi(i + 1), \ldots, \pi(j)\}$.
\term{A monotonic subsequence} is a subsequence whose entries form a monotonic sequence. 
\term{A run} is a maximal monotonic consecutive subsequence. 
For example, $(4, 5, 3, 1, 2)$ contains a (decreasing) run of length~$3$. 

The key notion of a permutation pattern is introduced below:
\begin{definition}
A permutation $\sigma$ on the set $[l]$ is \term{a pattern of a permutation $\pi$} on the set $[n]$ if there exists an increasing function $\phi : [l] \to [n]$ such that
$$\forall x, y \in [l] : \sigma(x) < \sigma(y) \text{ if and only if } \pi(\phi(x)) < \pi(\phi(y)).$$
We say that the function $\phi$ certifies the pattern. 
\end{definition}
A parameterized problem is a language $Q \ss \Sigma^*\times\N$, where $\Sigma$ is a fixed alphabet. 
The value $k$ of the instance $(x, k) \in Q$ is its \term{parameter}. 
\permpat\ is the following parameterized algorithmic problem:
\begin{framed}
\noindent
\textbf{Input:} a permutation $\sigma$ on $[\ell]$, a permutation $\pi$ on $[n]$. \\
\textbf{Parameter:} $\ell$. \\
\textbf{Question:} is $\sigma$ a pattern of $\pi$?
\end{framed}
A parameterized problem $Q$ is in \FPT\ if there is an algorithm deciding $(x, k) \in Q$ in time $f(k)|x|^{\O(1)}$, where $f$ is a computable function.
\begin{definition}
\term{A kernelization algorithm for a parameterized problem $Q$} is an algorithm that given an instance $(x, k) \in \Sigma^* \times \N$ produces in $p(|x| + k)$ steps an instance $(x', k')$ such that
\begin{enumerate}
\item $(x, k) \in Q \Leftrightarrow (x', k') \in Q$ and 
\item $|x'|, k' \le f(k)$, 
\end{enumerate} 
where $p(\cdot)$ is a polynomial and $f(\cdot)$ a computable function. 
\end{definition}
If there is a kernelization algorithm for $Q$, we say that \term{$Q$ has a kernel}. 
If the function $f(\cdot)$ in the above definition can be bounded by a polynomial, we say that \term{$Q$ has a polynomial kernel}. 

We utilize the standard machinery of Bodlaender et al. \cite{Bcross}, which builds on \cite{Cross1, Cross2}, to derive kernelization lower-bounds under the widely believed complexity assumption \assumptioncoNP.
A failure of this assumption would imply that the polynomial hierarchy collapses to the third level. 
Some basic definitions are necessary:
\begin{definition}[Bodlaender et al. \cite{Bcross}]
An equivalence relation $R$ on $\language$ is called \term{a polynomial equivalence relation} if the following two conditions hold:
\begin{enumerate}
\item There is an algorithm that given two strings $x, y \in \language$ decides whether $x$ and $y$ belong to the same equivalence class in $(|x| + |y|)^{\O(1)}$ time.
\item For any finite set $S \in \language$ the equivalence relation $R$ partitions the elements of $S$ into at most $(\max_{x \in S}|x|)^{\O(1)}$ equivalence classes.
\end{enumerate}
\end{definition}
An example of such a relation is the grouping of instances of the same size.
\begin{definition}[Bodlaender et al. \cite{Bcross}]
Let $L \ss \language$ be a language and let $Q \ss \language \times \N$ be a parameterized problem.
We say that \term{$L$ cross-composes into $Q$} if there is a polynomial equivalence relation $R$ and an algorithm which, given $t$ strings $x_1, x_2, \ldots, x_t$ belonging to the same equivalence class of $R$, computes an instance $(x^*, k^*) \in \language \times \N$ in time polynomial in $\sum_{i=1}^t |x_i|$
such that:
\begin{enumerate}
\item $(x^*, k^*) \in Q \Leftrightarrow x_i \in L$ for some $1 \le i \le t,$
\item $k^*$ is bounded by a polynomial in $\max_{i=1}^t |x_i| + \log t.$
\end{enumerate}
\end{definition}
\begin{theorem}[Bodlaender et al. \cite{Bcross}]
\label{t.cross} 
Let $L \ss \Sigma^*$ be an \NP-hard language.
If $L$ cross-composes into a parameterized problem $Q$ and $Q$ has a polynomial kernel then $\mbox{\NP} \subseteq \mbox{co-\NP}/\mbox{poly}$.
\end{theorem}

\section{Kernelization lower bound for \permpat}

We show that the \permpat\ problem is unlikely to have a polynomial kernel: 

\begin{theorem}
\label{t.main}
Unless $\mbox{\NP} \subseteq \mbox{co-\NP}/\mbox{poly}$, the \permpat\ problem does not have a polynomial kernel.
\end{theorem}

We prove Theorem \ref{t.main} using Theorem \ref{t.cross}. 
However, this requires a polynomial-time reduction that allows cross-composition without significantly increasing the parameter value.
Reductions described in the literature \cite{Lnp, M} have resisted our attempts to apply the framework. 
Therefore, we introduce a new \NP-hardness proof that directly leads to a cross-composition. 
The new reduction is from the well known \clique\ problem. 

Let us first define encoding $\pi_z(G)$ taking a graph $G$ and $z \in \N$ and producing a permutation. 
The key property of the encoding is that for any $G, H$ we have $G \ss H$ if and only if $\pi_z(G)$ is a pattern of $\pi_z(H)$ for a particular choice of $z$.
This allows us to express the \clique\ problem in terms of \permpat. 
The encoding permutation itself consists of two types of entries: encoding entries and separating entries. 
The former ones encode the edges of $G$. 
The latter form decreasing runs used to separate encoding entries corresponding to different vertices. 
The mapping $\pi_z(\cdot)$ can be seen as an embedding of the upper-triangular submatrix of the adjacency matrix of the input graph into a permutation. 
The separating runs mark where each row and column begins and ends; 
the encoding entries determine where the 1-entries of the matrix are. 
This mapping, along with some notation introduced below, is illustrated in Figure \ref{f.encoding}.

We start constructing $\pi_z(G)$ by imposing a total order on $V(G)$ placing vertices from the same connected component of $G$ consecutively.
Thus, we can assume that $V(G) = [n]$ and set
\begin{align*}
N^+_G(v)    & := \{ u : u > v \land \{u, v\} \in E(G) \}, \\
N^-_G(v)    & := \{ u : u < v \land \{u, v\} \in E(G) \}, \\
\deg^+_G(v) & := |N^+_G(v)|, \\ 
\deg^-_G(v) & := |N^-_G(v)|.
\end{align*}
We call the vertices from $N^+_G(v)$ and $N^-_G(v)$ the right-neighbours and left-neighbours of~$v$, respectively. 
The overall structure of the permutation $\pi_z(G)$ is as follows. 
It starts with a decreasing run of length $z$,
continues with the entries encoding the vertex 1 (i.e., encoding $N^+_G(1)$), 
which is then followed by another decreasing run of length $z$. 
This finishes the part of the permutation dedicated to the vertex 1 and the segment for the vertex 2 begins.
Again, it starts with another decreasing run of length $z$, continues with the encoding entries of $N^+_G(2)$, and is finished by a decreasing run of length $z$. 
This continues for all vertices of $G$. 
Note that for each vertex $v$ there is a pair of decreasing runs immediately surrounding the entries encoding $N^+_G(v)$, one from left and one from right. 
These are called the left and right separating runs of $v$, respectively. 
\begin{figure}[b!]
\centering
\begin{tikzpicture}[scale=0.25,style=thick]

\draw[red!70!black,line width=2pt] (8.5, -1) -- (8.5, 29);
\draw[red!70!black,line width=2pt] (13.5, -1) -- (13.5, 29);
\draw[red!70!black,line width=2pt] (-1, 26.5) -- (29, 26.5);
\draw[red!70!black,line width=2pt] (-1, 21.5) -- (29, 21.5);

\foreach \x/\y in
{{0/5}, {1/4}, {2/3}, {4/2}, {5/1}, {6/0},   {7/12}, {8/11}, {9/10}, {12/8}, {13/7}, {14/6},   {15/19}, {16/18}, {17/17}, {19/15}, {20/14}, {21/13},   {22/27}, {23/26}, {24/25}, {25/22}, {26/21}, {27/20}} 
	\fill[gray!70!white] (\x, \y) rectangle (\x + 1, \y + 1);

\foreach \x/\y in {{3/9}, {10/16}, {11/23}, {18/24}}
	\fill[red!70!black] (\x, \y) rectangle (\x + 1, \y + 1);

\draw (0, 0) grid (28, 28);

\draw[->] (-1, -2.5) node[below=-3] {\footnotesize $p_L(1)$} to[out=60,in=-90,distance=30] (0.5, -0.5);
\draw[->] (3.5, -4.5) node[below=-3] {\footnotesize $p_M(1)$} to (3.5, -0.5);
\draw[->] (6, -2.5) node[below=-3] {\footnotesize $p_R(1)$} to[out=135,in=-90,distance=30] (4.5, -0.5);
\draw[->] (9.5, -4.5) node[below=-3] {\footnotesize $p_L(2)$} to[out=125,in=-90,distance=30] (7.5, -0.5);
\draw[->] (10.5, -2.5) node[below right=-3] {\footnotesize $p_M(2)\hspace{1mm}\cdots$} to[out=135,in=-90,distance=30] (10.5, -0.5);

\matrix [column sep=0,shift={(-4,3.3)},left delimiter={(},right delimiter={)}]
{
\node {0}; & \node{1}; & \node{1}; & \node {0}; \\
\node {0}; & \node{1}; & \node{0}; & \node {}; \\
\node {1}; & \node{0}; & \node{}; & \node {}; \\
\node {0}; & \node{}; & \node{}; & \node {}; \\
};

\node[single arrow, fill=white, anchor=base, draw=black, align=center, text width=1cm] at (-5.4, 13.14) {};

\end{tikzpicture}
\caption[]{
Left part shows the upper triangular submatrix of the adjacency matrix of a graph $G = \big(\{1,2,3,4\},\big\{\{1,2\}, \{2,3\}, \{2,4\}, \{3,4\}\big\}\big)$. 
The right part shows the permutation matrix representation of its encoding permutation $\pi_3(G)$.
The columns of both matrices are indexed from left to right, the rows from bottom to top. 
Thus, the $(1,1)$ entry of the matrix is in the bottom-left corner.
White positions of the grid on the right correspond to 0-entries of the permutation matrix, non-white positions are 1-entries. 
Separating runs are colored in light gray, encoding entries in dark red.
Note the one-to-one correspondence between the 1-entries of the matrix on the left with the encoding entries of the permutation matrix. 
Horizontal red lines represent \textit{values} attained at positions $L(4)$ and $R(4)$. 
Vertical red lines denote \textit{indexes} $L(2)$ and $R(2)$. 
Note that these four red lines induce a rectangle with a single 1-entry. 
This encodes the 1-entry in the top-most position of the second column of the adjacency matrix.
Arrows below the permutation matrix illustrate the notation $p_L(\cdot), p_M(\cdot),$~and~$p_R(\cdot)$. 
}
\label{f.encoding}
\end{figure}
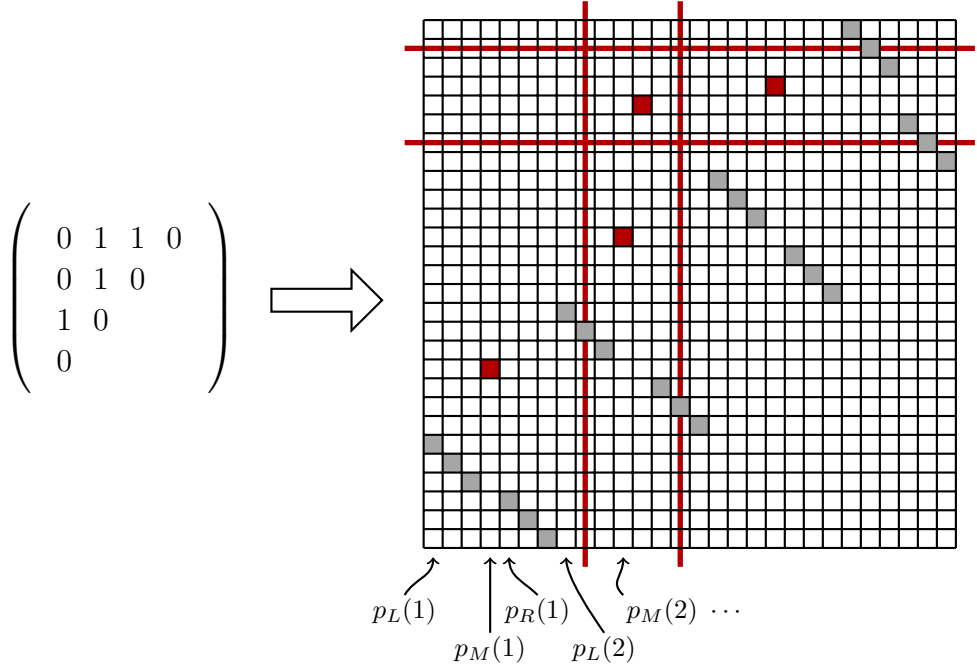


To facilitate the formal definition of $\pi_z(G)$, 
we begin by introducing a notation for important positions and values of the resulting permutation's entries. 
This includes the positions where the abovementioned runs start, the values with which they start, or the positions where the parts encoding $N^+_G(v)$, for individual choices of $v$, start. 

We use $p_L(v)$ and $p_R(v)$ as a shorthand for the positions on which the left and right separating run of $v$ starts, respectively. 
The first position of the segment encoding $N^+_G(v)$ is denoted by $p_M(v)$. 
This is illustrated in Figure \ref{f.encoding}. 
Specifically, we set $p_L(1) := 1, p_M(1) := z + 1,$ and $p_R(1) := z + 1 + \deg^+_G(1)$. 
For $v \ge 2$, we have: 
\begin{align*} 
p_L(v)      & := p_R(v - 1) + z, \\
p_M(v)      & := p_L(v) + z, \\
p_R(v)      & := p_M(v) + \deg^+_G(v). 
\end{align*} 
We also introduce notation for the \textit{values} used by the separating runs. 
The left separating run of $v$ starts at the position $p_L(v)$ with the value $q_L(v)$. 
The right separating run starts at $p_R(v)$ with the value $q_R(v)$. 
Finally, $q_M(v)$ is the least value used 
by the encoding entries of vertices from $N_G^-(v)$ to determine their connection to $v$. 
(Specifically, vertices of $N_G^-(v)$ use the values $[q_M(v), q_M(v) + \deg_G^-(v) - 1]$ to encode this.
If $\deg^-_G(v)$ is zero, the value $q_M(v)$ is actually not used.) 
We set $q_L(1) := 2z, q_M(1) := z + 1,$ and $q_R(1) := z$. 
For $v \ge 2$, let
\begin{align*} 
q_R(v)      & := q_L(v - 1) + z, \\
q_M(v)      & := q_R(v) + 1, \\
q_L(v)      & := q_M(v) + z + \deg^-_G(v) - 1. 
\end{align*} 
We now define the values of $\pi = \pi_z(G)$.
For each $v$, we introduce a decreasing run of length $z$ starting at the position $p_L(v)$: 
\begin{align*} 
\pi(p_L(v)) & := q_L(v), \\
\pi(p_L(v) + 1) & := q_L(v) - 1, \\
\pi(p_L(v) + 2) & := q_L(v) - 2, \\
\ldots \\ 
\pi(p_L(v) + z - 1) & := q_L(v) - (z - 1).
\end{align*}
We also insert a decreasing run which starts at the position $p_R(v)$ with the value $q_R(v)$: 
\begin{align*} 
\pi(p_R(v)) & := q_R(v), \\
\pi(p_R(v) + 1) & := q_R(v) - 1, \\
\pi(p_R(v) + 2) & := q_R(v) - 2, \\
\ldots \\ 
\pi(p_R(v) + z - 1) & := q_R(v) - (z - 1).
\end{align*} 
This establishes the entries represented by gray squares in Figure \ref{f.encoding}. 

The remaining values are used to encode the edges of $G$. 
The neighbourhood $N^+_G(v)$ is encoded by an increasing run on positions $p_M(v), p_M(v) + 1, \ldots, p_M(v) + |N^+_G(v)| - 1$.
We fix a vertex $v \in V(G)$ and iterate through the neighbours $\{u_1, u_2, \ldots, u_k \} = N^+_G(v)$. 
Assume $u_1 < u_2 < \ldots < u_k$. 
For $i \in [k]$, we set: 
\begin{align} 
\label{eq.j}
\pi(p_M(v) + i - 1) := q_M(u_i) + \ell(v, u_i),
\end{align} 
where $\ell(v, u_i) = |\{ w : w < v \land \{w, u_i\} \in E(G) \}|$.
The term $\ell(v, u_i)$ ensures that no value in $\pi$ is repeated. 

The above procedure is carried out for each $v \in V(G)$. 
This finishes the construction of $\pi_z(G)$. 
We now provide two observations. 

\begin{observation} 
For any graph $G$ and $z \in \N$ the function $\pi_z(G)$ is a permutation. 
\end{observation} 

\begin{proof}
Let $\pi := \pi_z(G)$. 
It is straightforward to verify that $\pi$ is a mapping from $[p]$ to $[p]$, for $p = 2zn + |E(G)|$. 
It remains to show that $\pi$ is injective, i.e. 
that there is no pair of distinct indexes $i, j$ such that $\pi(i) = \pi(j)$. 
It is easily seen that such $i$ and $j$ cannot both be an index of an entry forming a separating run, since the separating runs are explicitly constructed so that the sets of their values are disjoint. 
For each vertex $v$ there are exactly $\deg^-_G(v)$ values between the values of its left and right separating run. 
These values are used to encode the $\deg^-_G(v)$ edges connecting $v$ to its left-neighbours. 
The left-most neighbour is using the least value, the subsequent vertices are using values that increase by 1 with each neighbour (cf. the term $\ell(v, u_i)$ in equation (\ref{eq.j})). 
Therefore, we have neither a collision between a separating entry and an encoding entry nor a collision between two entries encoding $N^-_G(v)$ for the same $v$. 
Finally, it can be easily seen that the sets of values encoding $N^-_G(v)$ are pairwise disjoint for different choices of $v$.
\end{proof}

\begin{observation} 
\label{o.ones}
For any choice of $z \in \N$ and any choice of $u, v \in V(G)$, 
there is at most one 1-entry of $\pi_z(G)$ with an index in $[p_M(u), p_R(u) - 1]$ and value in $[q_M(v), q_M(v) + \deg^-_G(v) - 1]$. 
Furthermore, 
there is an edge between vertices $u, v \in V(G), u < v$ if and only if there is exactly one such 1-entry. 
\end{observation} 

\begin{proof}
The entries of $\pi := \pi_z(G)$ with indexes from $[p_M(u), p_R(u) - 1]$ encode the neighbourhood of the vertex $u$, i.e., $N^+_G(u)$. 
For each neighbour $v$ from $N^+_G(u)$, we insert a single entry with value from $[q_M(v), q_M(v) + \deg^-_G(v) - 1]$. 
This is because in equation (\ref{eq.j}) the term $\ell(v, u_i)$ is always strictly less than $\deg^-_G(v)$. 
This implies both parts of the observation. 
\end{proof}

For the purpose of the proof of the lemma below, we define the following: 
\begin{align*}
C(v) & := [p_M(v), p_R(v) - 1], \\ 
L(v) & := p_L(v) + \lfloor \tfrac{z}{2} \rfloor, \\ 
R(v) & := p_R(v) + \lfloor \tfrac{z}{2} \rfloor. 
\end{align*}
Therefore, $C(v)$ is the set of entries of $\pi$ encoding the vertex $v$, 
$L(v)$ denotes the middle entry of the left separating run of $v$, 
and $R(v)$ is the middle entry of the right separating run of $v$. 
Once more, Figure \ref{f.encoding} illustrates the notation. 

The following lemma implies \NP-hardness of \permpat:
\begin{lemma} 
\label{l.key}
For every graph $G$ without isolated vertices, $K_l$ is a subgraph of $G$ if and only if $\pi_z(K_l)$ is a pattern of $\pi_z(G)$, for $z = 4n' + 4,$ where $n'$ is the number of vertices in the largest connected component of $G$. 
\end{lemma} 

\begin{proof}
We let $\sigma := \pi_z(K_l)$ and $\pi := \pi_z(G)$. 

If $G$ contains a clique $K_l$ of size $l$ as a subgraph, then $\pi$ contains the pattern $\sigma$ by construction. 
This is because if we consider the permutation matrix representation of $\pi$ and delete all columns except the ones that correspond to separating and encoding entries for vertices of $K_l \ss G$, we get a matrix that differs from the permutation matrix of $\sigma$ only by the additional presence of columns that encode the connection of the vertices of $K_l$ to the vertices outside of $K_l$. 
By deleting these columns (and the empty rows resulting from the above deletions) we arrive at the permutation matrix representation of $\sigma$, implying $\sigma$ is a pattern of $\pi$. 

For the other direction, assume there is a function $\phi: [|\sigma|] \to [|\pi|]$ certifying the pattern. 
We start by noting that there are no decreasing subsequences of length $\frac{1}{4}z$ in $\pi$ avoiding all separating runs. 
This is because such a sequence contains at most one entry from $C(v)$ for each $v \in V(G)$. 
At the same time, it cannot simultaneously contain an entry from $C(u)$ and $C(v)$ for $u, v$ chosen from different connected components. 
This is because the construction of the encoding permutation places vertices from the same component consecutively and the entries encoding a component placed earlier in the ordering have strictly smaller values than those from a later component. 
This bounds the length of the subsequence by $n' < \frac{1}{4}z$.

Furthermore, {\em any} decreasing subsequence of $\pi$ contains entries from at most one pair of separating runs. 
This is because once the sequence hits a separating run of a vertex $v$, all its subsequent entries can only be from the pair of separating runs of $v$. 
Any decreasing subsequence therefore starts with less than $\frac{1}{4}z$ encoding entries, which are then followed by entries of a pair of separating runs of some vertex. 

We now show that the certifying function $\phi$ naturally leads to a mapping from $V(K_l)$ to $V(G)$. 
Consider any vertex $v \in V(K_l)$. 
The function $\phi$ maps the subsequence of $\sigma$ formed by the pair of separating runs of $v$ to a decreasing subsequence of $\pi$ of the same length. 
As argued, such a long decreasing subsequence starts with less than $\frac{1}{4}z$ encoding entries of $\pi$, which are then followed by at least $\frac{7}{4}z$ entries from a pair of separating runs of some vertex $u \in V(G)$.  
This implies that the middle entry $L(v)$ of the left separating run of $v$ in $\sigma$ needs to be mapped by $\phi$ to the left separating run of $u \in G$. 
Additionally, the middle entry $R(v)$ of the right separating run of $v$ in $\sigma$ needs to be mapped by $\phi$ somewhere in the right separating run of the same vertex $u$.
The above establishes a mapping from $V(K_l)$ to $V(G)$ denoted by $f_\phi$. 

We claim $f_\phi$ to be a graph homomorphism. 
Fix any pair of vertices $v_1, v_2$ of $K_l$ such that $v_1 < v_2$. 
We show that $f_\phi(v_1), f_\phi(v_2)$ are connected by an edge in $G$. 
Since there is an edge between $v_1$ and $v_2$ in $K_l$, Observation \ref{o.ones} implies that the set of values $\sigma[L(v_1), R(v_1)]$ contains precisely one number $p$ with $\sigma(R(v_2)) \le p \le \sigma(L(v_2))$. 
Since $\phi$ certifies the pattern $\sigma$ in $\pi$, there needs to be an entry of $\pi$ with an index between $\phi(L(v_1))$ and $\phi(R(v_1))$ and value between $\pi(\phi(R(v_2)))$ and $\pi(\phi(L(v_2)))$. 
Observation \ref{o.ones} then implies there is an edge between $f_\phi(v_1)$ and $f_\phi(v_2)$. 
Thus, $f_\phi$ is a homomorphism and $G$ contains a clique of size $l$. 
%
\end{proof}

The above reduction can be directly used within the cross-composition framework to show our result: 

\begin{proof}[Proof of Theorem \ref{t.main}]
We set $L$ to be the set of all pairs $(K_l, G)$, where $K_l$ is a clique, $G$ is a connected graph containing $K_l$ as a subgraph. 
It is widely known that deciding $x \in L$ is \NPc. 

We introduce a cross-composition of $L$ into \permpat. 
Let $R$ be an equivalence relation on $\{0, 1\}^*$ with the following properties:  
the binary sequences that are not representing a pair $(K_l, G)$, where $K_l$ is a clique and $G$ a graph, are placed in a single equivalence class designated for malformed input sequences;
a pair of strings representing instances $(K^1, G^1)$ and $(K^2, G^2)$, respectively, is related in $R$ if and only if $|V(K^1)| = |V(K^2)|, |V(G^1)| = |V(G^2)|$. 
Clearly, $R$ is a polynomial equivalence relation. 
For instances $(K_l, G_1), (K_l, G_2), (K_l, G_3), \ldots, (K_l, G_t)$ from the same equivalence class of $R$, we produce an instance of the \permpat\ problem where we ask if $\pi_z(K_l)$ is in $\pi_z(G)$, where $G$ is a disjoint union of graphs $G_1, \ldots, G_t$ and $z$ is set to $4 \cdot |V(G_1)| + 4$. 
Lemma \ref{l.key} shows that the answer to this problem is YES if and only if at least one of the instances $(K_l, G_i)$ belongs to $L$. 
Since the parameter of the pattern matching instance is $|\pi_z(K_l)|$, which can be bounded by $|V(G_i)| \cdot 2z + |V(G_i)|^2$ for any $i$, we can apply Theorem \ref{t.cross}. 
\end{proof}

\section{Conclusion} 

Guillemot and Marx \cite{Glin} have shown that the \permpat\ problem can be solved in $2^{\O(\ell^2 \log \ell)} \cdot n$ time. 
They raised the question of whether a faster \FPT\ algorithm could be obtained and outlined a strategy for achieving this using their notion of decompositions of permutations. 
This relied on the bound from the Stanley-Wilf conjecture not being tight. 
Fox \cite{F} has shown this is not the case. 
(Still, \cite{F} gives an improved $2^{\O(\ell^2)} \cdot n$ algorithm.) 
The non-existence of a polynomial kernel is a further indication of the difficulty of the problem. 

\bibliographystyle{siam}
\bibliography{all}

\end{document}